\documentclass[11pt]{amsart}
\usepackage{amsmath, amssymb}
\usepackage{graphics,color}

\newtheorem{theorem}{Theorem}[section]

\newtheorem{lemma}[theorem]{Lemma}

\def\IC{{\mathbb C}}
\def\IR{{\mathbb R}}

\def\cS{{\mathcal S}}

\def\cU{{\mathcal U}}
\def\cP{{\mathcal P}}

\def\({\left (}
\def\){\right )}

\def\Diag{{\rm Diag}\,}
\def\tr{{\rm tr}}
\def\rank{{\rm rank}\,}
\def\span{{\rm span}\,}
\def\conv{{\rm conv}\,}
\def\Id{{\rm Id}}

\def\H{H_{n_1\cdots n_m}}
\def\M{M_{n_1\cdots n_m}}
\def\A{A_1\otimes\cdots\otimes A_m}
\def\m{\widetilde{m}}
\def\n{\widetilde{n}}

\begin{document}
\openup .8\jot

%%%%%%%%%%%%%%%%%%%%%%%%%%%%%%%%%%%%%%%%%%%%%%%%%%%%%%%%%%%%%%%%%%%%%%%%%%%%%%%%%%%%%%%%%%%%%%%%%%%%%%%%%%%%%%%%%%%%%%%%%%%%%%%%%%%%%%%%%%%%%%%%%%%%%%%%%%%%%%%%%%%%%%%%%%%%%%%%%%%%%%%%%%%%%%%%%%%%%%%%%%%%%%%%%%%%%%%%%%%%%%%%%%%%%%%%%%%%%%%%%%%%%%%%%%%%%%%%%%%%%%%%%%%%%%%%%%%%%%%%%%%%%%%%%%%%%%%%%%%%%%%%%%%%%%%%%%%%%%%%%%%%%%%%%%%%%%%%%%%%%%%%%%%%%%%%%%%%%%%%%%%

\title[Linear preservers and quantum information science]
      {Linear preservers and \\ quantum information science}

\author{Ajda Fo\v sner}
\author{Zejun Haung}
\author{Chi-Kwong Li}
\author{Nung-Sing Sze}
\address{Ajda Fo\v sner,
Faculty of Management, University of Primorska,
Cankarjeva 5, SI-6104 Koper, Slovenia}
\email{ajda.fosner@fm-kp.si}
\address{Zejun Haung,
Department of Applied Mathematics,
The Hong Kong Polytechnic University,
Hung Hom, Hong Kong}
\email{huangzejun@yahoo.cn}
\address{Chi-Kwong Li,
Department of Mathematics, College of William and Mary, Williamsburg, VA 23187, USA;
Department of Mathematics,
University of Hong Kong, Pokfulam, Hong Kong}
\email{ckli@math.wm.edu}
\address{Nung-Sing Sze,
Department of Applied Mathematics,
The Hong Kong Polytechnic University,
Hung Hom, Hong Kong}
\email{raymond.sze@polyu.edu.hk}

\date{}

\maketitle

\centerline{\em Dedicated to Professor Pjek-Hwee Lee on the occasion of his retirement.}

\begin{abstract}
In this paper, a brief survey of recent results on linear preserver problems
and quantum information science is given. In addition, characterization is
obtained for linear operators $\phi$ on $mn \times mn$ Hermitian matrices
such that $\phi(A\otimes B)$ and $A\otimes B$ have the same spectrum
for any $m\times m$ Hermitian $A$ and $n\times n$ Hermitian $B$. Such a map
has the form $A\otimes B \mapsto  U(\varphi_1(A) \otimes  \varphi_2(B))U^*$
for $mn\times mn$ Hermitian matrices in tensor form $A\otimes B$, where
$U$ is  a unitary matrix, and for $j \in \{1,2\}$,
$\varphi_j$ is the identity map $X \mapsto X$ or the transposition map $X \mapsto X^t$.
The structure of linear maps leaving invariant the spectral radius
of matrices in tensor form $A\otimes B$ is also obtained. The results are connected
bipartite (quantum) systems and are extended to multipartite systems.
\end{abstract}
\medskip
{\em 2010 Math. Subj. Class.}: 15A69, 15A86, 15B57, 15A18.


%%%%%%%%%%%%%%%%%%%%%%%%%%%%%%%%%%%%%%%%%%%%%%%%%%%%%%%%%%%%%%%%%%%%%%%%%%%%%%%%%%%%%%%%%%%%%%%%%%%%%%%%%%%%%%%%%%%%%%%%%%%%%%%%%%%%%%%%%%%%%%%%%%%%%%%%%%%%%%%%%%%%%%%%%%%%%%%%%%%%%%%%%%%%%%%%%%%%%%%%%%%%%%%%%%%%%%%%%%%%%%%%%%%%%%%%%%%%%%%%%%%%%%%%%%%%%%%%%%%%%%%%%%%%%%%%%%%%%%%%%%%%%%%%%%%%%%%%%%%%%%%%%%%%%%%%%%%%%%%%%%%%%%%%%%%%%%%%%%%%%%%%%%%%%%%%%%%%%%%%%%%
\section{Introduction}
%%%%%%%%%%%%%%%%%%%%%%%%%%%%%%%%%%%%%%%%%%%%%%%%%%%%%%%%%%%%%%%%%%%%%%%%%%%%%%%%%%%%%%%%%%%%%%%%%%%%%%%%%%%%%%%%%%%%%%%%%%%%%%%%%%%%%%%%%%%%%%%%%%%%%%%%%%%%%%%%%%%%%%%%%%%%%%%%%%%%%%%%%%%%%%%%%%%%%%%%%%%%%%%%%%%%%%%%%%%%%%%%%%%%%%%%%%%%%%%%%%%%%%%%%%%%%%%%%%%%%%%%%%%%%%%%%%%%%%%%%%%%%%%%%%%%%%%%%%%%%%%%%%%%%%%%%%%%%%%%%%%%%%%%%%%%%%%%%%%%%%%%%%%%%%%%%%%%%%%%%%%

The study of linear preserver problems has a long history.
It concerns the characterization of linear maps on matrices or operators with special
properties. For example, Frobenius \cite{F} showed that a linear operator
$\phi: M_n\rightarrow M_n$ satisfies
$$\det(\phi(A)) = \det(A) \qquad \hbox{ for all } A \in M_n$$
if and only if  there are $M, N\in M_n$ with $\det(MN) = 1$ such that
$\phi$ has the form
\begin{equation} \label{standard}
A \mapsto MAN \quad \hbox{ or } \quad A \mapsto MA^tN,
\end{equation}
where $M_n$ denotes the set of $n\times n$ complex matrices.
Clearly, a map of the form (\ref{standard}) is linear and leaves the determinant
function invariant. It is interesting that a linear map preserving the determinant
function must be of this form. In \cite{D}  
Dieudonn\'{e} showed that an invertible linear operator $\phi: M_n \rightarrow M_n$
maps the set of singular matrices into itself if and only if there are invertible 
$M,N \in M_n$
such that $\phi$ has the form (\ref{standard}).
One may see \cite{LP} and its references for results on linear preserver problems.
There are many new directions and active research on preserver problems motivated
by theory and applications; see \cite{B,M7,W}.

In this paper, we focus on linear preserver problems related to quantum information
science.  In Section 2, we briefly survey some recent results on such research,
and motivate our study in Section 3, in which we characterize linear preservers
of the spectral radius or the spectrum of the tensor product of two Hermitian matrices, and
discuss the implications of the result to bipartite quantum systems.
The results are extended to the tensor product of $m$ Hermitian matrices
with $m > 2$ corresponding to the multipartite quantum systems.
Additional remarks, results and open problems are also presented.

\section{Quantum information science and preservers}

Let $H_n$ be the set of Hermitian matrices in $M_n$.
In quantum physics, {\em quantum states} of a system with $n$
physical states are represented as
{\rm density matrices} $A$ in $H_n$, i.e., $A$ is
positive semi-definite with trace one.
Rank one orthogonal projections are pure states.

The classical Wigner's theorem in quantum mechanics 
asserts that a bijective map $\phi$ on the set of pure states
satisfying $\tr(AB) = \tr(\phi(A)\phi(B))$ must be of the form
\begin{equation}\label{standard2}
A \mapsto UAU^* \qquad \hbox{ or } \qquad A \mapsto UA^tU^*
\end{equation}
for some unitary operator $U$. 
Uhlhorn \cite{U} showed that a bijective map
$\phi$ on the set of pure states also has the form (\ref{standard2})
under the weaker assumption that  $\tr(AB) = 0$ if and only if $\tr \phi(A)\phi(B)) = 0$.  
The result was extended to  
Hilbert modules over matrix algebras, prime C*-algebras, and indefinite inner
product spaces; see
\cite{M1,M4}. In \cite{LPSe}, the authors extended Uhlhorn's result to
Hermitian matrices, symmetric matrices, the set of
orthogonal projections, the set of rank one orthogonal projections,
and the set of effect algebra, and studied bijective maps on these matrix sets
such that 
$$\tr(AB) = c \qquad\hbox{ if and only if } \qquad \tr(\phi(A)\phi(B)) = c$$
for a given $c > 0$.

In a series of interesting papers \cite{M2,M3,M4,M6,MB}, Moln\'{a}r and his collaborators
characterized bijective maps on the set of complex matrices, Hermitian matrices,
bounded observables,  effect algebra, etc. preserving special subsets or relations.
In many cases, the map has the form (\ref{standard2}).
One may see also \cite{M7} for additional results along this direction.

Suppose $A \in H_m$ and $B \in H_n$ are the states of two quantum systems.
Then the {\em tensor (Kronecker) state}  $A\otimes B \in H_{mn}$ 
describes the joint (bipartite) system.
A density matrix $C \in H_{mn}$ is {\em separable} if it is the convex combination of
tensor states, i.e., $C = \sum_{j=1}^r t_j A_j\otimes B_j$ for some positive
numbers $t_1, \dots, t_r$ summing up to one, and tensor states $A_1\otimes B_1,
\dots, A_r \otimes B_r$. Otherwise, $C$ is {\em entangled}. Identifying
separable states in $H_{mn}$ is an NP-hard problem; see \cite{G}.
Nevertheless, there is of interest in finding easy ways to check
necessary or sufficient conditions
of separability of states. 
In particular, it is interesting to find transformations which will simplify a given
state so that it is easier to determine whether it is separable or not. Evidently, 
the transformations used should not change the set of separable states.  This leads to
the study of linear operators leaving invariant the set of separable states (entangled states).
Similar definitions and questions can be considered for multipartite systems.
The following result was proved in \cite{FLPS}.

\begin{theorem} Let $n_1, \dots, n_m \in \{2, 3, \dots\}$ and $N = \prod_{j=1}^m n_j$.
Suppose $\cS$ is one of the following.

\medskip
{\rm (a)} The set of tensor product (of pure) states
$A_1 \otimes \cdots \otimes A_m$, where  $A_j \in H_{n_j}$ is a (pure) state
for each $j \in \{1, \dots, m\}$.

\medskip
{\rm (b)} The set of separable states in $H_N$, viz,
the convex hull of the set of tensor product (of pure) states.

\medskip\noindent
Then a linear map $\phi: H_{N} \rightarrow H_{N}$ satisfies $\phi(\cS) = \cS$
if and only if there is a permutation $(p_1, \dots, p_m)$ of $(1, \dots, m)$ such that
$$A_1 \otimes \cdots \otimes A_m \mapsto \psi_1(A_{p_1}) \otimes \cdots \otimes \psi_m(A_{p_m}),$$
where for each $j \in \{1, \dots, m\}$,
$n_j = n_{p_j}$ and $\psi_j: M_{n_j} \rightarrow M_{n_j}$ is a linear map
of the form
$$X \mapsto U_jXU_j^* \qquad \hbox{ or } \qquad X \mapsto U_jX^tU_j^*$$
for a unitary $U_j \in M_{n_j}$.
\end{theorem}

The result was generalized in three directions by researchers.
First, Hou and his associates \cite{hou}
extended the result to the infinite dimensional setting and characterized
bounded invertible linear maps leaving invariant the set of tensor product of
rank one orthogonal projections acting on infinite dimensional Hilbert spaces,
or its convex hull, i.e., the set of separable states. Second, Lim \cite{Lim}
characterized linear map $\phi: H_{n_1} \otimes \cdots \otimes H_{n_m} \rightarrow
H_{\tilde n_1} \otimes \cdots \otimes H_{\tilde n_m}$ such that $\phi$ maps the set of
tensor (separable) states in the domain into the set of tensor (separable) states
in the codomain. Third, the authors in \cite{LPS} characterize linear map
$\phi: H_{n_1} \otimes \cdots \otimes H_{n_m} \rightarrow
H_{n_1} \otimes \cdots \otimes H_{n_m}$
such that $\phi(\cS_1) = \cS_2$, where
$$\cS_1 = \{ X_1 \otimes \cdots \otimes X_m: X_j \in \cU(C_j), \ j = 1, \dots, m\}$$
and
$$\cS_2 = \{ Y_1 \otimes \cdots \otimes Y_m: Y_j \in \cU(D_j), \ j = 1, \dots, m\}$$
for given states
$C_j, D_j \in H_{n_j}$ with $j = 1, \dots, m$ and
$$\cU(X) = \{U^*XU: U \hbox{ unitary} \}$$
is the unitary (similarity) orbit of $X$. When $C_i$ and $D_i$ are pure states,
the study reduces to the problem treated in \cite{FLPS}, and reveals the fact that
there are linear transformations converting a unitary orbit to a different unitary orbit.

In \cite{J}, the author  showed a number of interesting linear preserver results related
to quantum information science. A vector state of a quantum system with $m$
measurable physical states can be represented as a unit vector $u$ in $\IC^m$.
A product state of two vector states $u \in \IC^m$ and $v \in \IC^n$
is the tensor product $u\otimes v \in \IC^{mn}$, and unit vectors in $\IC^{mn}$ can
be viewed as vector states in the bipartite system with $\IC^m$ and $\IC^n$ as
components.
Every vector $w \in \IC^{mn}$ can be identified
with an $m\times n$ matrix $[w]$ by putting the first $n$ entries in the first row,
the next $n$ entries in the second row, etc. In particular, $u \otimes v$ can be
identify with the matrix $uv^t$.  The singular value decomposition of
the matrix $[w] = \sum_{j=1}^k s_j u_jv_j^t$
corresponds to the Schmidt decomposition $w = \sum_{j=1}^k s_j u_j \otimes v_j$.
The Schmidt rank of a vector (state) $w$ is the rank of the matrix $[w]$.
Clearly, the linear span of product states $u\otimes v$
will generate all the vectors in $\IC^{mn}$, and a linear map $L$ on $\IC^{mn}$ is
completely determined once we know $L(u\otimes v)$ for all (or $mn$ linearly independent)
product states $u\otimes v$.
In \cite{J}, the author used some classical
results on linear preservers to study maps preserving $\cP_k$,
the set of all states with Schmidt rank at most $k$
for a given $k \le \min\{m,n\}$. In particular, it was shown that
an invertible linear map $L:\IC^{mn}\rightarrow \IC^{mn}$
satisfies $L(\cP_k) \subseteq \cP_k$
if and only if there are unitary matrices $P \in M_m$ and $Q \in M_n$
such that one of the following holds.

\medskip
(a) $L(u \otimes v) =  Pu \otimes Qv $ for all $(u,v) \in \IC^m\times \IC^n$.

(b) $m=n$ and $L(u\otimes v) = Qv \otimes Pu$ for all $(u,v) \in \IC^m\times \IC^n$.

\medskip\noindent
Suppose $\cS_k$ is the set of all vectors $w\in \IC^{mn}$ with Schmidt rank at most $k$.
Then an invertible linear map $L:\IC^{mn}\rightarrow \IC^{mn}$
satisfies $L(\cS_k) \subseteq \cS_k$
if and only if there are invertible matrices $P \in M_m$ and $Q \in M_n$
such that (a) or (b) holds.

\medskip
%For $u,v \in \IC^{mn}$, one can consider $uv^t \in M_{mn}$.
Another result in \cite{J} asserts that an invertible linear map $\Phi: M_{mn} \rightarrow M_{mn}$
satisfies $\Phi(\cS) \subseteq \cS$, where $\cS$ is the set of rank one matrices of
the form $uv^t$ such that $u$ and $v$ have Schmidt rank at most $k$
if and only if $\Phi$ is a composition of
one or more of the following maps.

\medskip

(1) The transpose map $X \mapsto X^t$.

(2) $X \mapsto (P_1\otimes Q_1)X(P_2\otimes Q_2)$ for some invertible matrices $P_i \in M_m$
and $Q_i \in M_n$ for $i = 1,2$.

(3) $k = 1$, the partial transpose map $\left[ X_{ij} \right]_{1\le i,j \le m}
\mapsto [ X_{ij}^t ]_{1 \le i,j \le m}$, where $X_{ij} \in M_n$.

\medskip
Furthermore, Johnston considered the norm on $\IC^{mn}$ defined by
$$\|u\|_{k} = \max\{ |v^*u|: v \in \IC^{mn},\ v^*v = 1,\ \rank([v]) \le k \} = \left\{\sum_{j=1}^k s_j^2\right\}^{1/2},$$
where  $s_1 \ge s_2 \ge \cdots $ are the singular values of $[u]$, for any $k \le \min \{m,n\}$.
He also considered the norm on $M_{mn}$ defined by
$$|||C|||_k = \max \left\{ |u^*Cv|: u, v \in \IC^{mn},\ u^*u = v^*v = 1,\ \rank([u]) \le k, \
\rank([v]) \le k \right\}.$$
These norms have recently been studied in \cite{5,9,10,11,18} and were shown to be related to
the problem of characterizing $k$-positive linear maps and detecting bound entangled
non-positive partial transpose states.

\medskip
In connection to the preserver problems, it was shown that a linear map
$L: \IC^{mn} \rightarrow \IC^{mn}$ satisfies
$$\|L(u)\|_{k} = \|u\|_k \qquad \hbox{ for all } u \in \IC^{mn}$$
if and only if there are unitary $P \in M_m$ and $Q \in M_n$ such that
condition (a) or (b) mentioned above holds.

\medskip
If $k = \min\{m,n\}$ one sees that $|||C|||_k$ is just the operator norm. It is known that
a linear preserver on $M_{mn}$ of the operator norm has the form
$$X \mapsto UXV \quad \hbox{ or } \quad X \mapsto UX^tV$$
for some unitary $U, V \in M_{mn}$.
For $k < \min\{m,n\}$, Johnston showed that a linear map $\Phi: M_{mn} \rightarrow M_{mn}$
satisfies
$$|||\Phi(X)|||_k = |||X|||_k \quad \hbox{ for all } X \in M_{mn}$$
if and only if $\Phi$ is a composition of
one or more of the maps described in (1), (2) or (3) above with the additional
restriction that $P$ and $Q$ in (2) are unitary.

\medskip
Many of the above results are extended to multi-partite system, e.g., \cite{FLPS,J,LPS, Lim}.

\medskip
Next, we consider another line of research in preserver problems.
There has been considerable interest in studying spectrum preserving maps (see
\cite{CLS,Sourour,Marcus} etc).
On Hermitian matrices, it is known that a linear map on $H_n$ that leaves invariant
the spectrum has the form
$$A \mapsto UAU^* \quad \hbox{ or } \quad A \mapsto UA^tU^*$$
for some unitary $U \in M_n$. If one gives up the Hermitian preserving property
and considers a (complex) linear operator
$\phi: M_n \rightarrow M_n$ that leaves invariant the eigenvalues
of Hermitian matrices,  then $\phi$ has the form
\begin{equation} \label{sas}
A \mapsto SAS^{-1} \quad \hbox{ or } \quad A \mapsto SA^tS^{-1}
\end{equation}
for some invertible $S \in M_n$.

In \cite{Saitoh1,Saitoh}, the authors studied non-classical correlation in a bipartite
systems and showed that for any spectrum preserving linear
map $\phi: H_n \rightarrow M_n$,
either
$$\sigma((\Id_m\otimes \phi)(C))= \sigma(C) \quad \hbox{ for all }\quad C \in H_m\otimes H_n,$$
or
$$\sigma((\Id_m\otimes \phi)(C) = \sigma({\rm PT}_2(C)) \quad \hbox{ for all }\quad C \in H_m\otimes H_n,$$
where ${\rm PT}_2(A\otimes B) = A \otimes B^t$
is the partial transpose map for the second component and $\Id_m$ is the identity map on $m\times m$ matrices.

Following this line of study, we consider linear operators leaving invariant
the spectrum of tensor states and related problems in the next section.
It turns out that even if one assumes only that a linear operator $\phi$
leaves invariant the spectrum of matrices in tensor form $A\otimes B\in H_m\otimes H_n$,
the operator $\phi$ has a nice structure, namely,
up to a unitary similarity, $\phi$ has the form
$A\otimes B \mapsto \psi_1(A) \otimes \psi_2(B)$ for all tensor states $A\otimes B$,
where $\psi_j$ is the identity map $X\mapsto X$ or the transposition map $X \mapsto X^t$.
Moreover, if $\sigma(C) = \sigma(\phi(C))$ for a carefully chosen $C \in H_{mn}$,
then $\phi$ will actually preserve the spectrum of every matrix in $H_{mn}$, and will
be of the form $X \mapsto VXV^*$ or $X \mapsto VX^tV^*$ on $H_{mn}$ for some unitary
matrix $V \in H_{mn}$. Similar results are obtained for linear maps leaving invariant
the spectral radius of tensor states $A\otimes B$ in $H_m\otimes H_n$.

\section{Preservers of spectral radius or spectrum}

Suppose $A\in H_m$ has eigenvalues $a_1 \ge \cdots \ge a_m$ associated with orthonormal
eigenvectors $x_1, \dots, x_m$, and $B\in H_n$ has eigenvalues
$b_1 \ge \dots \ge b_n$ associated with orthonormal eigenvectors
$y_1, \dots, y_n$, then  $A\otimes B$ has eigenvalues $a_rb_s$ associated with
eigenvectors $x_r\otimes y_s$ for $(r,s) \in \{1, \dots, m\}\times \{1, \dots, n\}$.
Denote by $\sigma(X)$ and $r(X)$ the spectrum and spectral radius of
a matrix $X \in M_n$.
In Subsection 3.1, we show that a linear map
$\phi: H_m \otimes H_n \rightarrow H_m\otimes H_n$ satisfies
$$\sigma(\phi(A\otimes B)) = \sigma(A\otimes B)$$
for all $A\otimes B \in H_m \otimes H_n$ if and only if
there is a unitary $U\in M_{mn}$ such that
\begin{equation}\label{S1}
A\otimes B \mapsto U( \varphi_1(A) \otimes \varphi_2(B) )U^*,
\end{equation}
where $\varphi_j$, $j=1,2$, is either the identity map or
the transposition map $X \mapsto X^t$ (see Theorem \ref{2.1}).
Furthermore, we will also show that a linear map on $H_{mn}$ leaving the spectral radius of tensor states invariant, i.e.,
$$r(\phi(A\otimes B)) = r(A\otimes B)$$ for all $A\otimes B \in H_m \otimes H_n$,
is $\pm 1$ multiple of a map of the standard form (\ref{S1}) (see Theorem \ref{2.2}).
In Subsection 3.2,
we will extend the results to multipartite systems (Theorem \ref{3.1} and Theorem \ref{3.2}).
Additional remarks, results, and open problems  will be presented in Subsection 3.3.

%%%%%%%%%%%%%%%%%%%%%%%%%%%%%%%%%%%%%%%%%%%%%%%%%%%%%%%%%%%%%%%%%%%%%%%%%%%%%%%%%%%%%%%%%%%%%%%%%%%%%%%%%%%%%%%%%%%%%%%%%%%%%%%%%%%%%%%%%%%%%%%%%%%%%%%%%%%%%%%%%%%%%%%%%%%%%%%%%%%%%%%%%%%%%%%%%%%%%%%%%%%%%%%%%%%%%%%%%%%%%%%%%%%%%%%%%%%%%%%%%%%%%%%%%%%%%%%%%%%%%%%%%%%%%%%%%%%%%%%%%%%%%%%%%%%%%%%%%%%%%%%%%%%%%%%%%%%%%%%%%%%%%%%%%%%%%%%%%%%%%%%%%%%%%%%%%%%%%%%%%%%

\iffalse

Similar to the study in \cite{Marcus}, we consider also results on complex matrices.
We will characterize linear maps on $M_m\otimes M_n \equiv M_{mn}$ that leave
invariant the spectrum or spectral radius of matrices in tensor form
$A \otimes B \in M_m\otimes M_n$ (see Theorem \ref{3.1} and Theorem \ref{3.2}).

We will present some preliminary results in Section 2.
The results on Hermitian matrices will be treated in Section 3 and
those on complex matrices in Section 4.  Extension to
multipartite systems and further research will be described in
Section 5.

\fi

%%%%%%%%%%%%%%%%%%%%%%%%%%%%%%%%%%%%%%%%%%%%%%%%%%%%%%%%%%%%%%%%%%%%%%%%%%%%%%%%%%%%%%%%%%%%%%%%%%%%%%%%%%%%%%%%%%%%%%%%%%%%%%%%%%%%%%%%%%%%%%%%%%%%%%%%%%%%%%%%%%%%%%%%%%%%%%%%%%%%%%%%%%%%%%%%%%%%%%%%%%%%%%%%%%%%%%%%%%%%%%%%%%%%%%%%%%%%%%%%%%%%%%%%%%%%%%%%%%%%%%%%%%%%%%%%%%%%%%%%%%%%%%%%%%%%%%%%%%%%%%%%%%%%%%%%%%%%%%%%%%%%%%%%%%%%%%%%%%%%%%%%%%%%%%%%%%%%%%%%%%%
\subsection{Bipartite system}
%%%%%%%%%%%%%%%%%%%%%%%%%%%%%%%%%%%%%%%%%%%%%%%%%%%%%%%%%%%%%%%%%%%%%%%%%%%%%%%%%%%%%%%%%%%%%%%%%%%%%%%%%%%%%%%%%%%%%%%%%%%%%%%%%%%%%%%%%%%%%%%%%%%%%%%%%%%%%%%%%%%%%%%%%%%%%%%%%%%%%%%%%%%%%%%%%%%%%%%%%%%%%%%%%%%%%%%%%%%%%%%%%%%%%%%%%%%%%%%%%%%%%%%%%%%%%%%%%%%%%%%%%%%%%%%%%%%%%%%%%%%%%%%%%%%%%%%%%%%%%%%%%%%%%%%%%%%%%%%%%%%%%%%%%%%%%%%%%%%%%%%%%%%%%%%%%%%%%%%%%%%
Throughout this paper, we denote by $E_{ij}, 1\leq i,j\leq n$ the standard basis of $M_n$.
We need the following lemma.

\begin{lemma} \label{le3.1}
Let $m>n$ and $A\in H_m$ with  $\sigma(A)=\{a_1,\ldots, a_n,0,\ldots,0\}$. If
$$\sigma(A+t(I_n\oplus 0_{m-n}))=\{a_1+t,\ldots,a_n+t,0,\ldots,0\} \textrm{ for all } t\in \mathbb{R},$$
then $A=B\oplus 0_{m-n}$ for some $B\in H_n$.
\end{lemma}

\begin{proof}
Choose a sufficient large $s\in \IR$ so that $C = A + s(I_n \oplus 0_{m-n})$
is positive semi-definite with eigenvalues $c_1,\dots,c_n,0,\dots,0$
where $c_j = a_j + s$, $j = 1,\dots, n$. Then
$$\sigma(C + t (I_n \oplus 0_{m-n})) = \sigma(A + (s+t)(I_n \oplus 0_{m-n}))
= \{c_1+t,\dots,c_n+t,0,\dots,0\}.$$
Denote by $\{e_1,\dots,e_m\}$ the standard basis of $\IC^m$.
Then for any unit vector $v \in \span \{e_{n+1},\dots,e_m\}$,
$$v^* C v
= v^* (C+t(I_n \oplus 0_{m-n})) v
\in \conv \{c_1 + t,\dots,c_n +t,0\}\quad \hbox{for all } \quad t\in \IR,$$
where $\conv S$ denote the convex hull of the set $S$.
Since this holds for all $t$ in $\IR$,
this is possible only when $v^*C v = 0$.
As $C$ is positive semi-definite,
$v$ is an eigenvector of $C$ with eigenvalue $0$.
As $v$ is arbitrary in $\span \{e_{n+1},\dots,e_m\}$,
$C$ must have the form $C_1 \oplus 0_{m-n}$. Hence,
$A = B \oplus 0_{m-n}$ with $B = C_1 - sI_n$.
\end{proof}

\iffalse
We prove by induction on $n$.
Denote by $C=A+t(I_n\oplus 0_{m-n})$. Without loss of generality, we assume that $t$ is sufficient large and that $C$ has a unique maximal eigenvalue $a_1+t$. Then $A$ and $I_n\oplus 0_{m-n}$ has a common unit eigenvector $x=(x_1,\ldots,x_n,0,\ldots,0)^T$ such that $x^*Cx=a_1+t$.

If $n=1$, one can verify that $A=a_1xx^*=[a_1] \oplus 0_{m-1}$.
Assume $n\geq 2$ and the statement holds for $n-1$.
Choose a unitary $U$ of the form $U_1 \oplus U_2$ with $U_1 \in H_n$ and $U_2 \in H_{m-n}$
so that $x$ is the first column of $U$.
Then
 $$U^*AU=\left[\begin{array}{cc}
              a_1&0\\
              0&A_1
              \end{array}
              \right]~~~\textrm{and}~~~U^*(I_n\oplus 0_{m-n})U=\left[\begin{array}{cc}
              1&0\\
              0&I_{n-1}\oplus 0_{m-n}
              \end{array}
              \right].$$
 It follows that
 $$\sigma(A_1+t(I_{n-1}\oplus 0_{m-n}))=\{a_2+t,\ldots,a_n+t,0,\ldots,0\} \textrm{ for all } t\in \mathbb{R}.$$
By the induction hypothesis we have $A_1=B_1 \oplus 0_{m-n}$ for some $B_1\in H_{n-1}$. Hence,
$$A=U\left[\begin{array}{cc}
              a_1&0\\
              0&A_1
              \end{array}\right]U^*=B\oplus 0_{m-n}$$
 where $B=U_1(a_1\oplus B_1)U_1^*\in H_n$. Thus the statement of Lemma \ref{le3.1} is true for $n$.
\end{proof}
\fi

\begin{theorem}
\label{2.1}
A linear map $\phi: H_{mn}\rightarrow H_{mn}$ satisfies
$$\sigma(\phi(A\otimes B)) = \sigma(A\otimes B)$$
for all $A\otimes B \in H_m\otimes H_n$
if and only if there is a unitary $U \in M_{mn}$ such that
$$\phi(A\otimes B)=U(\varphi_1(A)\otimes \varphi_2(B))U^*,$$
where  $\varphi_j$ is the identity map or  the transposition map $X \mapsto X^t$
for $j \in \{1, 2\}$.
\end{theorem}

\begin{proof}
The sufficiency part is clear. We consider the necessity part.
Since $\sigma(\phi(I_m\otimes I_n)) = \sigma(I_m\otimes I_n) = \{1\}$, we see that
$\phi(I_m\otimes I_n) = I_m\otimes I_n$.
Consider any distinct pairs $(j,k)$ and $(r,s)$ for $j,r\in \{1,\ldots, m\}$,
$k,s\in \{1,\ldots, n\}$. Then $\phi(E_{jj}\otimes E_{kk})$ and $\phi(E_{rr}\otimes E_{ss})$
are nonzero orthogonal projections. Now, $I_{mn} = \phi(I_{mn}) = \sum_{j,k}\phi(E_{jj}\otimes E_{kk})$ has trace $mn$.
It follows that each $\phi(E_{jj}\otimes E_{kk})$ has rank one. Moreover,  $\phi(E_{jj} \otimes E_{kk})$ and $\phi(E_{rr} \otimes E_{ss})$
have disjoint range spaces for any distinct pairs $(j,k)$ and $(r,s)$. Hence, there exists a unitary $W\in M_{mn}$ such that
$$\phi(E_{jj} \otimes E_{kk}) = W ( E_{jj} \otimes E_{kk} )W^*$$ for all $1\le j \le m$ and $1\le k \le n$.

For any $B \in H_n$, $t \in \IR$, and $1\le j \le m$,  we have
\begin{eqnarray*} &&
\sigma \left(  \phi(E_{jj} \otimes B) + t \phi(E_{jj} \otimes I_n) \right)
= \sigma \left( \phi( E_{jj} \otimes (B+t I_n)) \right) \\
&=& \sigma \left( E_{jj} \otimes (B+t I_n) \right)
= \{b + t: b \in \sigma(B)\} \cup \{0\}.
\end{eqnarray*}
Since $\phi(E_{jj} \otimes I_n) = W(E_{jj} \otimes I_n)W^*$, applying Lemma \ref{le3.1} and using  permutation  similarity if necessary,
we have
$$\phi(E_{jj} \otimes B) = W(E_{jj} \otimes \psi_j(B))W^*$$
for some $\psi_j(B) \in H_n$. Furthermore, $B$ and $\psi_j(B)$
have the same spectrum. So $\psi_j$ has the form
$$B \mapsto   U_jBU_j^* \quad \hbox{ or } \quad B \mapsto U_jB^tU_j^* $$
for some unitary $U_j$.
Replace $W$ with $W(U_1\oplus\cdots \oplus U_m)$. Then
\begin{equation}\label{eq2}
\phi(E_{jj} \otimes B)= W(E_{jj} \otimes \varphi_j(B))W^*
\end{equation}
for all $1\le j \le m$ and $B \in H_n$,  where
each map $\varphi_j$ is the identity map or  the transposition map $X \mapsto X^t$.

Repeating the same argument, one can show that for any unitary $U \in M_m$,
$$\phi(UE_{jj}U^* \otimes B)= W_U(E_{jj} \otimes \varphi_{j,U}(B))W_U^*$$
for all $1\le j \le m$ and $B \in H_n$,
where $W_U\in M_{mn}$ is a unitary matrix, depending on $U$, and $\varphi_{j,U}$ is either the identity map or the transposition map,
depending on $j$ and $U$. Replacing $\phi$ by the map $A\mapsto W_{I_{mn}}^*\phi(A)W_{I_{mn}}$,
we may assume that $$W_{I_{mn}}= I_{mn} \quad {\rm and} \quad \phi(E_{jj}\otimes E_{kk})= E_{jj}\otimes E_{kk}$$
for all $1\le j \le m$ and $1\le k \le n$.
Now, for any real symmetric $S\in H_n$ and unitary $U \in M_m$, we have
$\varphi_{j,U}(S) = S$ for all $j = 1, \dots, m$, and, hence,
$$\phi(I_m \otimes S)
= \phi\left( \sum_{j=1}^m UE_{jj}U^* \otimes S \right)
= W_U \left(\sum_{j=1}^m E_{jj} \otimes S\right)W_U^*
= W_U \left( I_m \otimes S \right)W_U^*$$
for some unitary $W_U\in M_{mn}$. In particular, when $U = I_m$,
$\phi(I_m \otimes S) = I_m \otimes S$.
Thus, $W_U \left( I_m \otimes S \right)W_U^* = I_m \otimes S$.
It follows that
$W_U$ commutes with $I_m \otimes S$ for all real symmetric $S$.
Hence, $W_U$ has the form $V_U \otimes I_n$ for some $V_U \in M_m$ and
\begin{equation}\label{eq4}
\phi\left(UE_{jj}U^* \otimes B \right) = (V_UE_{jj}V_U^*) \otimes \varphi_{j,U}(B)
\end{equation}
for $1\le j \le m$ and $B \in M_m$.
Consider the linear maps $\tr_1: H_{mn} \to H_n$
and  $\Phi: H_{mn} \to H_n$ defined by
$$\tr_1(A\otimes B) = (\tr A)B \quad \hbox{ and } \quad
\Phi(A\otimes B) = \tr_1\left( \phi(A\otimes B) \right)$$
for any $A\otimes B \in H_m\otimes H_n.$
Then
\begin{equation*}
\Phi\left(UE_{jj}U^* \otimes B \right) = \varphi_{j,U}(B).
\end{equation*}
Recall that  a continuous image of a connected space is still connected. Since $\Phi$ is linear and continuous, $\{xx^*\in M_m :x^*x=1\}$ is connected, and $\varphi_{j,U}$ is either the identity map or the transposition map,
all the maps $\varphi_{j,U}$ have to be the same.
Replacing $\phi$ by the map  $A\otimes B \mapsto \phi(A\otimes B^t)$, if necessary,
we may assume that this common map is the identity map.
%$$\phi \left( UE_{jj}U^* \otimes B \right)
%= (V_U UE_{jj} U^* V_U^*) \otimes B.$$
Next, by linearity, one can conclude that
for every $A \in H_m$ and $B \in H_n$ we have
$$\phi \left(A \otimes B \right) = \varphi_1(A) \otimes B$$
for some $\varphi_1(A) \in H_m$, where $\varphi_1(A)$ depends on $A$ only.
Note that $\varphi_1:H_m\to H_m$ is a linear map and $\sigma(\varphi_1(A)) = \sigma(A)$ for all $A \in H_m$.
Hence, by \cite{Marcus}, a map  $\varphi_1$ has the form $A\mapsto VAV^*$ or $A\mapsto VA^tV^*$.
The proof is completed.
\end{proof}

%%%%%%%%%%%%%%%%%%%%%%%%%%%%%%%%%%%%%%%%%%%%%%%%%%%%%%%%%%%%%%%%%%%%%%%%%%%%%%%%%%%%%%%%%%%%%%%%%%%%%%%%%%%%%%%%%%%%%%%%%%%%%%%%%%%%%%%%%%%%%%%%%%%%%%%%%%%%%%%%%%%%%%%%%%%%%%%%%%%%%%%%%%%%%%%%%%%%%%%%%%%%%%%%%%%%%%%%%%%%%%%%%%%%%%%%%%%%%%%%%%%%%%%%%%%%%%%%%%%%%%%%%%%%%%%%%%%%%%%%%%%%%%%%%%%%%%%%%%%%%%%%%%%%%%%%%%%%%%%%%%%%%%%%%%%%%%%%%%%%%%%%%%%%%%%%%%%%%%%%%%%

In the following, we consider linear maps on $H_{mn}$ leaving the spectral radius invariant.

\begin{theorem}
\label{2.2}
A linear map $\phi: H_{mn}\rightarrow H_{mn}$ satisfies
$$r(\phi(A\otimes B)) = r(A\otimes B)$$
for all $A\otimes B \in H_m\otimes H_n$
if and only if there is a unitary $U \in M_{mn}$ and $\lambda \in \{-1,1\}$ such that
$$\phi(A\otimes B)= \lambda U(\varphi_1(A)\otimes \varphi_2(B))U^*,$$
where  $\varphi_j$ is the identity map or the transposition map $X \mapsto X^t$
for $j \in \{1, 2\}$.
\end{theorem}

\begin{proof}
The sufficiency part is clear. For the converse, suppose that a linear map $\phi: H_{mn}\rightarrow H_{mn}$
preserves the spectral radius of tensor states and let $1\le j \le m$, $1\le k \le n$.
Then $\phi(E_{jj}\otimes E_{kk})$ has an eigenvalue in $\{1, -1\}$.
For $t \ne k$, we have $r(\phi(E_{jj} \otimes (E_{kk} \pm E_{tt}))) = 1$. This yields that every eigenvector of $\phi(E_{jj}\otimes E_{kk})$
corresponding to the eigenvalue $1$ or $-1$ lies in the kernel of $\phi(E_{jj} \otimes E_{tt})$. Since this is true for any pair of $k$ and $t$,
for any orthogonal diagonal matrix $D \in M_n$ at least $n$ eigenvalues of $\phi(E_{jj}\otimes D)$ lie in $\{1,-1\}$.
%The sufficiency part is clear. For the converse, suppose that a li\-ne\-ar map $\phi: H_{mn}\rightarrow H_{mn}$
%preserves the spectral radius of tensor states and let $1\le j \le m$, $1\le k \le n$. Then $\phi(E_{jj}\otimes E_{kk})$ has the eigenvalue
%in $\{1, -1\}$. For $r \ne k$ we have $r(\phi(E_{jj} \otimes (E_{kk} \pm E_{rr}))) = 1$. This yields that the eigenvectors of $\phi(E_{jj}\otimes %E_{kk})$
%corresponding to the eigenvalues $1$ and $-1$ lie in the kernel of $\phi(E_{jj} \otimes E_{rr})$. Since this is true for any $k = 1, \dots, n$,
%there is a unitary $W_j \in M_{mn}$ and a diagonal orthogonal matrix $P_{j,D}\in H_n$ such that $$W_j\phi(E_{jj}\otimes D)W_j^* = P_{j,D} \oplus %\psi_j(D)$$
%for any diagonal orthogonal matrix $D \in H_n$, where $\psi_j(D)\in H_{(m-1)n}$.
Since $r(\phi((E_{jj} \pm E_{ss})\otimes D)) = 1$ for any $j \ne s$, $1\le j,s\le m$, and any diagonal orthogonal matrix $D \in H_n$,
$\phi(E_{jj} \otimes D)$ and $\phi(E_{ss} \otimes D)$ have disjoint support and, hence, $\phi(E_{jj} \otimes D)$ has rank $n$.
It follows that all $\phi(E_{jj} \otimes E_{kk})$ must be rank one and
$\phi(E_{jj} \otimes E_{kk})$ and $\phi(E_{ss} \otimes E_{tt})$ have disjoint support
for any distinct $(j,k)$ and $(s,t)$.
Therefore, there is a unitary $W \in M_{mn}$ and $\mu_{jk} \in \{1, -1\}$ such that
$$\phi(E_{jj}\otimes E_{kk}) = \mu_{jk} W(E_{jj}\otimes E_{kk})W^*
\quad \hbox{for all } \quad 1\le j \le m, 1\le k \le n.$$
For the sake of the simplicity, suppose that $W = I_{mn}$
and
$\phi(E_{jj} \otimes I_n) = E_{jj}\otimes P_j$,
where $P_1, \dots, P_m\in H_n$ are diagonal orthogonal matrices.

For any unitary $V \in M_n$, applying the same arguments to $E_{jj}\otimes VE_{kk}V^*$, $1\le j \le m$, $1\le k \le n$,
we see that  $\phi(E_{jj} \otimes VE_{kk}V^*)$ has rank one with spectral radius $1$. If  $t > 0$, we have
$$r(\phi(E_{jj}\otimes (VE_{kk}V^* + tI_n))) = 1+t.$$ Thus, the eigenspace of the nonzero eigenvalue of $\phi(E_{jj} \otimes VE_{kk}V^*)$
must lie in the eigenspace of $\phi(E_{jj}\otimes I_n) = E_{jj}\otimes P_j$. Consequently, we see that $\phi(E_{jj}\otimes B) = E_{jj} \otimes \varphi_j(B)$
for any $B\in H_n$. Clearly, $\varphi_j$ preserves spectral radius on $H_n$ and, hence, by \cite{LC} it has the form
$$B \mapsto \xi YBY^* \quad \hbox{ or } \quad B \mapsto \xi YB^tY^*$$ for some $\xi \in \{1,-1\}$ and unitary $Y \in M_n$.
In particular, $\varphi_j(I_n) \in \{I_n, -I_n\}$. So, $\phi(I_{mn}) = D \otimes I_n$ for some diagonal orthogonal matrix $D \in M_m$.

By considering $UE_{jj}U^*\otimes E_{kk}$ for unitary $U\in M_m$ and using the same arugment
as in the last paragraph, one can show that $\phi(I_{mn}) = I_m \otimes \tilde D$ for some diagonal orthogonal matrix $\tilde D \in M_n$.
%Now, consider $U^*E_{jj}U\otimes E_{kk}$ for unitary $U\in M_m$ and $1\le j \le m$, $1\le k \le n$. If $t > 0$, then
%$r(\phi((U^*E_{jj}U + tI_m)\otimes E_{kk}) = 1+t$. It follows that $\phi(I_{mn}) = I_m\otimes \tilde D$ for some diagonal orthogonal matrix $\tilde D \in H_n$.
Since $\phi(I_{mn}) = I_m\otimes \tilde D = D \otimes I_n$, we conclude that  $\phi(I_{mn})=\pm I_{mn}$.
Without loss of generality, we may assume that $\phi(I_{mn})=I_{mn}$.
Thus, all $\mu_{jk}$ are equal to $1$, i.e.,
$$\phi(E_{jj} \otimes E_{kk}) = E_{jj} \otimes E_{kk}
\quad \hbox{for all}\quad  1 \le j \le m, 1\le k \le n.$$
%By Theorem \ref{2.1}, it suffices to show that $\sigma(\phi(A\otimes B)) = \sigma(A\otimes B)$ for all $A\otimes B \in H_m\otimes H_n$.

For any $A\otimes B \in H_m \otimes H_n$, there are unitary $U \in M_m$ and $V \in M_n$ such that $UAU^*$ and $VBV^*$ are diagonal matrices.
Without loss of generality, we assume that
$A = \Diag(a_1,\dots,a_m)$ and
$B = \Diag(b_1,\dots,b_n)$.
Then
%Using the previous argument, there exist a unitary $W\in M_{mn}$ and $\mu_{jk} \in \{1,-1\}$ such that
%$$\phi(E_{jj} \otimes E_{kk} ) = \mu_{jk} W( E_{jj} \otimes E_{kk})W^*.$$ Since $\phi(I_{mn}) = I_{mn}$, $\mu_{jk} = 1$ for all $1\le j \le m$, $1\le k \le n$, and
$$\phi(A\otimes B) = \phi\left( \left(\sum_{j=1}^m a_j E_{jj} \right) \otimes \left(\sum_{k=1}^n b_k E_{kk}\right)\right)= A\otimes B.$$
Thus, $\sigma(\phi(A\otimes B)) = \sigma(A\otimes B)$ and the result is followed by Theorem \ref{2.1}.
%as desired. The proof is completed.
\end{proof}

%%%%%%%%%%%%%%%%%%%%%%%%%%%%%%%%%%%%%%%%%%%%%%%%%%%%%%%%%%%%%%%%%%%%%%%%%%%%%%%%%%%%%%%%%%%%%%%%%%%%%%%%%%%%%%%%%%%%%%%%%%%%%%%%%%%%%%%%%%%%%%%%%%%%%%%%%%%%%%%%%%%%%%%%%%%%%%%%%%%%%%%%%%%%%%%%%%%%%%%%%%%%%%%%%%%%%%%%%%%%%%%%%%%%%%%%%%%%%%%%%%%%%%%%%%%%%%%%%%%%%%%%%%%%%%%%%%%%%%%%%%%%%%%%%%%%%%%%%%%%%%%%%%%%%%%%%%%%%%%%%%%%%%%%%%%%%%%%%%%%%%%%%%%%%%%%%%%%%%%%%%%
\subsection{Multipartite systems}
%%%%%%%%%%%%%%%%%%%%%%%%%%%%%%%%%%%%%%%%%%%%%%%%%%%%%%%%%%%%%%%%%%%%%%%%%%%%%%%%%%%%%%%%%%%%%%%%%%%%%%%%%%%%%%%%%%%%%%%%%%%%%%%%%%%%%%%%%%%%%%%%%%%%%%%%%%%%%%%%%%%%%%%%%%%%%%%%%%%%%%%%%%%%%%%%%%%%%%%%%%%%%%%%%%%%%%%%%%%%%%%%%%%%%%%%%%%%%%%%%%%%%%%%%%%%%%%%%%%%%%%%%%%%%%%%%%%%%%%%%%%%%%%%%%%%%%%%%%%%%%%%%%%%%%%%%%%%%%%%%%%%%%%%%%%%%%%%%%%%%%%%%%%%%%%%%%%%%%%%%%%

In this section we will extend Theorem \ref{2.1} and Theorem \ref{2.2} to multipartite system $H_{n_1\cdots n_m}=H_{n_1}\otimes\cdots \otimes H_{n_m}$, $m\ge 2$.

\begin{theorem}
\label{3.1}
A linear map $\phi:\H\rightarrow \H$ satisfies $$\sigma(\phi(\A)) = \sigma(\A)$$ for all $\A \in \H$
if and only if there is a unitary $U \in \M$ such that
\begin{equation}
\label{form1}
\phi(\A)=U(\varphi_1(A_1)\otimes \cdots\otimes\varphi_m(A_m))U^*,
\end{equation}
where  $\varphi_j$ is the identity map or  the transposition map $X \mapsto X^t$ for $j \in \{1, \ldots, m\}$.
\end{theorem}

\begin{proof}
The sufficiency part is clear. To prove the necessity part, we use induction on $m$. By Theorem \ref{2.1}, we already know that the statement
of Theorem \ref{3.1} is true for bipartite systems. So, assume that $m\ge 3$ and that the result holds for all $(m-1)$-partite systems.
We would like to prove that the same is true for $m$-partite systems.

As in the proof of Theorem \ref{2.1}, we can show that there exists a unitary $W\in M_{n_1\cdots n_m}$ such that
$$\phi(E_{j_1j_1} \otimes \cdots \otimes E_{j_mj_m}) = W ( E_{j_1j_1} \otimes \cdots \otimes E_{j_mj_m})W^*$$
for all $1\le j_p \le n_p$ with $1\le p \le m$. 
Moreover, for any  $B\in H_{n_1}$ and $1\le j_p \le n_p$ with $2\le p \le m$, we have
$$\phi(B \otimes E_{j_2j_2} \otimes \cdots \otimes E_{j_mj_m} )
= W(\psi_{j_2,\dots,j_m} (B) \otimes E_{j_2j_2} \otimes \cdots \otimes E_{j_mj_m})W^*$$
for some $\varphi_{j_2,\dots,j_m}(B)\in H_{n_1}$. Then $B$ and $\varphi_{j_2,\dots,j_m}(B)$
have the same spectrum.
By the fact that $\varphi_{j_2,\dots,j_m}(E_{kk}) = E_{kk}$ for all $1\le k \le n_1$,
the map $\varphi_{j_2,\dots,j_m}$ can be assumed either the identity map or the transposition map.
By a similar argument, we can show that
$$\phi\left( B \otimes \left( \bigotimes_{p=2}^m U_p E_{j_p j_p} U_p^* \right)  \right)
= W_{U_2,\dots, U_m} \left(\varphi_{j_2,\dots,j_m}^{U_2,\dots,U_m}(B) \otimes E_{j_2j_2} \otimes \cdots \otimes E_{j_mj_m}
 \right)W_{U_2,\dots, U_m}^*$$
for all $B \in H_{n_1}$ and $1\le j_p \le n_p$ with $2\le p \le m$, 
where $W_{U_2,\dots, U_m}\in M_{n_1\cdots n_m}$ is a unitary matrix depending on $U_2,\dots,U_m$ only 
and $\varphi_{j_2,\dots,j_m}^{U_2,\dots,U_m}$ is either the identity map or the transposition map, depending on 
$j_2,\dots,j_m$ and $U_2,\dots,U_m$.
Replacing $\phi$ by the map $A\mapsto W_{I_{n_2},\dots,I_{n_m}}^*\phi(A)W_{I_{n_2},\dots,I_{n_m}}$,
we may assume that 
$$W_{I_{n_2},\dots,I_{n_m}} = I_{n_1\cdots n_m} \quad \hbox{and}\quad
\phi\left( E_{j_1j_1} \otimes \cdots \otimes E_{j_mj_m} \right)= E_{j_1j_1} \otimes \cdots \otimes E_{j_mj_m}$$
for all $1\le j_p \le n_p$ with $1\le p \le m$.
Again, considering all symmetric $S\in H_{n_1}$  as in the proof of Theorem \ref{2.1}, 
we can show that there exists $V_{U_2,\dots,U_m} \in M_{n_2\cdots n_m}$ such that
$$\phi\left(B \otimes \left( \bigotimes_{p=2}^{m} U_p E_{j_p j_p} U_p^* \right)\right) 
= \varphi_{j_2,\dots,j_m}^{U_2,\dots,U_m}(B) \otimes 
V_{U_2,\dots, U_m} \left( E_{j_2j_2} \otimes \cdots \otimes E_{j_mj_m} \right) V_{U_2,\dots, U_{m}}^* .$$

Using the trace function, we see that all the maps $\varphi_{j_2,\dots,j_m}^{U_2,\dots,U_m}$ have to be the same. 
Assume that this common map is equal to $\varphi$, which is either the identity map or the transposition map.
By linearity, one can conclude that for any $A = A_2\otimes \cdots \otimes A_m\in H_{n_2\cdots n_m}$ and $B \in H_{n_1}$,
$$\phi \left(B \otimes A_2\otimes \cdots \otimes A_m \right) = \varphi(B) \otimes \psi(A_2\otimes \cdots \otimes A_m) $$ 
for some $\psi(A) = \psi_1(A_2\otimes \cdots \otimes A_m) \in H_{n_2\cdots n_m}$, where $\psi(A)$ depends on $A$ only.
Note that $\psi:H_{n_2\cdots n_m}\to H_{n_2\cdots n_m}$ is a linear map and $\sigma(\psi(A)) = \sigma(A)$ for all $A \in H_{n_2\cdots n_m}$.
Hence, by induction hypothesis, $\phi$ has the form (\ref{form1}), as desired. The proof is completed.
\end{proof}

\begin{theorem}
\label{3.2}
A linear map $\phi:\H\rightarrow \H$ satisfies $$r(\phi(\A)) = r(\A)$$ for all $\A \in \H$
if and only if there is a unitary $U \in \M$ and $\lambda \in \{-1,1\}$ such that
\begin{equation}
\label{form2}
\phi(\A)=\lambda U(\varphi_1(A_1)\otimes \cdots\otimes\varphi_m(A_m))U^*,
\end{equation}
where  $\varphi_j$ is the identity map or  the transposition map $X \mapsto X^t$ for $j \in \{1, \ldots, m\}$.
\end{theorem}

\begin{proof}
The sufficiency part is clear. To prove the converse, 
by a similar argument as in Theorem \ref{2.2}, we can show that 
$\phi(E_{j_1j_1} \otimes \cdots \otimes E_{j_mj_m})$ has an eigenvalue in $\{1,-1\}$ for any index set $(j_1,\dots,j_m)$, where $1\le j_p \le n_p$ with $1\le p \le m$.
Next, one can show that 
for any orthogonal diagonal matrix $D_1 \in H_{n_1}$,
$\phi(D_1 \otimes E_{j_2j_2} \otimes \cdots \otimes E_{j_{m}j_{m}})$
has at least $n_1$ eigenvalues lying in $\{1,-1\}$.
Furthermore, for any orthogonal diagonal matrices $D_1 \in H_{n_1}$ and $D_2 \in H_{n_2}$,
$\phi(D_1 \otimes D_2 \otimes E_{j_3j_3} \otimes \cdots \otimes E_{j_{m}j_{m}})$
has at least $n_1n_2$ eigenvalues lying in $\{1,-1\}$.
Recurrently, one can show that 
for any orthogonal diagonal $D_p \in H_{n_p}$ with $1\le p \le m$,
$\phi(D_1 \otimes D_2 \otimes \cdots \otimes D_m)$
has $n_1n_2\cdots n_m$ eigenvalues lying in $\{1,-1\}$.
This is possible only when $\phi(E_{j_1j_1} \otimes \cdots \otimes E_{j_mj_m})$
is rank one and for any distinct index sets $(j_1,\dots,j_m)$ and $(k_1,\dots,k_m)$,
$\phi(E_{j_1j_1} \otimes \cdots \otimes E_{j_mj_m})$
and $\phi(E_{k_1k_1} \otimes \cdots \otimes E_{k_mk_m})$
have disjoint support. Therefore,
there is a unitary matrix $W \in M_{n_1\cdots n_m}$ and $\mu_{j_1,\dots,j_m} \in \{1,-1\}$ such that
$$\phi(E_{j_1j_1} \otimes \cdots \otimes E_{j_mj_m}) = \mu_{j_1,\dots,j_m} W (E_{j_1j_1} \otimes \cdots \otimes E_{j_mj_m}) W^*.$$
%For the sake of the simplicity, suppose that $W = I_{n_1\cdots n_m}$ and 
Suppose $P_{j_2,\dots,j_m}$ are diagonal orthogonal matrices such that
$$\phi(I_{n_1} \otimes E_{j_2j_2} \otimes \cdots \otimes E_{j_mj_m}) = 
W(P_{j_2,\dots,j_m} \otimes  E_{j_2j_2} \otimes \cdots \otimes E_{j_mj_m})W^*.$$
Since every rank one matrix $R \in H_{n_1}$ 
can be expressed as $UE_{11}U^*$
for some unitary $U \in M_{n_1}$,
using the same argument as above, one can show that
$\phi(R \otimes E_{j_2j_2} \otimes \cdots \otimes E_{j_mj_m})$
has rank one with spectral radius $1$ for all $1\le j_p\le n_p$ with $2 \le p \le m$. By considering
$$r (\phi(  (R + t I_{n_1}) \otimes E_{j_2j_2} \otimes \cdots \otimes E_{j_mj_m} )) = 1 + t \quad\hbox{for all } t > 0,$$
one can conclude that 
$\phi(  R \otimes E_{j_2j_2} \otimes \cdots \otimes E_{j_mj_m} ) = W(\psi_{j_2,\dots,j_m}(R) \otimes E_{j_2j_2} \otimes \cdots \otimes E_{j_mj_m})W^*$
and hence for any $B \in H_{n_1}$,
$$\phi(  B \otimes E_{j_2j_2} \otimes \cdots \otimes E_{j_mj_m} ) = W(\psi_{j_2,\dots,j_m}(B) \otimes E_{j_2j_2} \otimes \cdots \otimes E_{j_mj_m})W^*.$$
Clearly, $\psi_{j_2,\dots,j_m}$ preserves spectral radius on $H_{n_1}$ and, hence, has the form $$B \mapsto \xi YBY^* \quad \hbox{ or } \quad B \mapsto \xi YB^tY^*$$
for some $\xi \in \{1,-1\}$ and unitary $Y \in M_{n_1}$. Then, one can see that the scalar $\mu_{j_1,\dots,j_m}$
has to be independent of the first index $j_1$, i.e., $\mu_{j_1,j_2\dots,j_m} =\mu_{j_1',j_2,\dots,j_m}$ for any $1\le j_1,j_1'\le n_1$.
Applying the same argument on the $p$th subsystem for $p = 2,\dots,m$, 
one can deduce that $\mu_{j_1,\dots,j_m}$ is independent of the $p$th index $j_p$.
Therefore, $\mu_{j_1,\dots,j_m} = \mu_{k_1,\dots,k_m}$ for any the index sets $(j_1,\dots,j_m)$ and $(k_1,\dots,k_m)$
and hence $\mu_{j_1,\dots,j_m} = \mu$ is a constant. So
$$\phi(E_{j_1j_1} \otimes \cdots \otimes E_{j_mj_m}) = \mu W (E_{j_1j_1} \otimes \cdots \otimes E_{j_mj_m})W^*
\quad \hbox{for all}\quad  1\le j_p\le m \hbox{ with } 1\le p \le m.$$
By the same argument, one can show that for any unitary $U_p \in M_{n_p}$ with $1 \le p \le m$,
$$\phi(U_1E_{j_1j_1}U_1^* \otimes \cdots \otimes U_mE_{j_mj_m}U_m^*) = \mu_{U_1,\dots,U_m} W_{U_1,\dots,U_m} (E_{j_1j_1} \otimes \cdots \otimes E_{j_mj_m}) W_{U_1,\dots,U_m}^*$$
for all $1\le j_p \le n_p$ with $1\le p \le m$. 
Here the scalar $\mu_{U_1,\dots,U_m} \in \{1,-1\}$ and the unitary matrix $W_{U_1,\dots,U_m} \in M_{n_1\cdots n_m}$ 
depend on $U_1,\dots,U_m$ only.
Furthermore, summing up for all the indices $j_1,\dots,j_m$ yields
$\phi(I_{n_1\cdots n_m}) = \mu_{U_1,\dots,U_m} I_{n_1\cdots n_m}$.
So $\mu_{U_1,\dots,U_m} = \mu_{I_{n_1},\dots,I_{n_m}} = \mu$ is independent of the choice of $U_1,\dots,U_m$.
Without loss of generality, we may assume that $\mu = 1$. 
Then by linearity, $\sigma(\phi(A_1\otimes \cdots\otimes A_m)) = \sigma(A_1\otimes \cdots\otimes A_m)$ for all $A_1\otimes \cdots\otimes A_m \in H_{n_1}\otimes\cdots\otimes H_{n_m}$, 
and the result follows from Theorem \ref{3.1}.
\end{proof}

\subsection{Additional remarks and results}

Several remarks concerning our results in the last two subsections are in order.

First, in all previous study of linear preservers involving
tensor product spaces, one always imposed the assumption that the preservers send tensor states
to tensor states. As a result, the structure of the preservers have the  form
\begin{equation} \label{tensorform}
A\otimes B \mapsto \psi_1(A) \otimes \psi_2(B) \quad \hbox{ or } \quad
A\otimes B \mapsto \psi_2(B) \otimes \psi_1(A).
\end{equation}
In our case, we do not assume that
the preservers send tensor states to tensor states. Nevertheless, our results show that
up to a unitary similarity, we still have the form (\ref{tensorform}).

Second, we characterize linear operators $\phi$ such that $A\otimes B$ and $\phi(A\otimes B)$
have the same spectrum (respectively, spectral radius). The resulting map may not preserve the
spectrum (respectively, spectral radius) of a general matrix $C \in H_{mn}$.
For example, if
$C = E_{11}\otimes E_{11} + E_{22}\otimes E_{22} + E_{12} \otimes E_{12} + E_{21} \otimes E_{21}$,
then the map $\phi$ of the form $A\otimes B \mapsto A \otimes B^t$ for tensor states will preserve
the spectral radius (and spectrum) of tensor states, but $\phi(C)$ and $C$ will not have
the same spectral radius (and spectrum). One can easily extend the above observation to
the following.

\begin{theorem} Suppose $\phi: H_{n_1\cdots n_m} \rightarrow
H_{n_1\cdots n_m}$ is linear such that
$r(\phi(C)) = r(C)$ (respectively, $\sigma(\phi(C)) = \sigma(C)$)
for all $C = A_1 \otimes \cdots \otimes A_m$
with $A_j \in H_{n_j}$, $j = 1, \dots, m$, and for
$C$ obtained from  $I_{n_1} \otimes \cdots \otimes I_{n_m}$ by replacing
$I_{n_i}\otimes I_{n_{i+1}}$ with
$E_{11}\otimes E_{11} + E_{22}\otimes E_{22} + E_{12} \otimes E_{12} + E_{21} \otimes E_{21}$,
$i = 1, \dots, m-1$.
Then there are a unitary $U$ and $\xi \in \{1, -1\}$ (respectively, $\xi = 1$)
such that  $\phi$ has the form
$$X \mapsto \xi UXU^* \quad \hbox{ or } \quad X \mapsto \xi UX^tU^*.$$
\end{theorem}

Third, one may consider affine maps $\psi$
on the set of density matrices in $H_N  = H_{n_1} \otimes \cdots \otimes H_{n_m}$
instead of linear maps on $H_N$. One may extend an
affine map on density matrices in $H_N$ in the standard way, namely, define
for any positive semi-definite matrix $C$,
$\phi(tC) = t\phi(C)$, and $\phi(C) = \psi(C)$ if $\tr C = 1$.
Then use the fact that every $X \in H_N$ is a difference of
two positive semi-definite $C_1$ and $C_2$,
and that $\phi(C_1) - \phi(C_2) = \phi(D_1) - \phi(D_2)$
if $C_1 - C_2 = D_1 - D_2$.

\iffalse
Fourth, can one modify our proof and prove the following.

\begin{theorem} Suppose $\phi: H_{n_1} \otimes \cdots \otimes H_{n_k}  \rightarrow
M_{n_1} \otimes \cdots \otimes M_{n_k}$ is linear such that
$r(\phi(C)) = r(C)$ (respectively, $\sigma(\phi(C)) = \sigma(C)$)
for all $C = A_1 \otimes \cdots \otimes A_k$
with $A_j \in H_{n_j}$, $j = 1, \dots, k$
Then there is a unitary $S$ and $\xi \in \{1, -1\}$ (respectively, $\xi = 1$)
such that  $\phi$ has the form
$$A_1 \otimes \cdots \otimes A_k \mapsto \xi S(\psi_1(A_1) \otimes \cdots \otimes \psi_k(A_k))S^{-1},$$
where $\psi_j(X)$ is the identity map or the transposition map for $j = 1, \dots, k$.
\end{theorem}
\fi

Finally, it is interesting to study (real or complex)
linear maps $\phi: M_{m} \otimes M_n \rightarrow M_{m}\otimes M_n$
such that $A\otimes B$ and $\phi(A\otimes B)$ always have the same spectrum (respectively,
spectral radius).

%%%%%%%%%%%%%%%%%%%%%%%%%%%%%%%%%%%%%%%%%%%%%%%%%%%%%%%%%%%%%%%%%%%%%%%%%%%%%%%%%%%%%%%%%%%%%%%%%%%%%%%%%%%%%%%%%%%%%%%%%%%%%%%%%%%%%%%%%%%%%%%%%%%%%%%%%%%%%%%%%%%%%%%%%%%%%%%%%%%%%%%%%%%%%%%%%%%%%%%%%%%%%%%%%%%%%%%%%%%%%%%%%%%%%%%%%%%%%%%%%%%%%%%%%%%%%%%%%%%%%%%%%%%%%%%%%%%%%%%%%%%%%%%%%%%%%%%%%%%%%%%%%%%%%%%%%%%%%%%%%%%%%%%%%%%%%%%%%%%%%%%%%%%%%%%%%%%%%%%%%%%
\bigskip
\noindent{\bf Acknowledgment}

%This research began when Fo\v{s}ner visited the Hong Kong Polytechnic University
%in the spring of 2011. She gratefully acknowledged the support and kind
%hospitality from the host university.
This research was supported by a Hong Kong RCG grant PolyU 502910 with Sze as PI and Li as co-PI.
The grant also supported the post-doctoral fellowship of Huang
and the visit of Fo\v{s}ner to the Hong Kong Polytechnic University in the spring of 2011. She gratefully acknowledged the support and kind hospitality from the host university.
Li was also supported by a USA NSF grant; this research was done when he was 
a visiting professor of the University of Hong Kong in the spring of 2012; furthermore,
he is an honorary professor of Taiyuan University of Technology (100 Talent Program scholar),
and an honorary professor of the  Shanghai University.

%%%%%%%%%%%%%%%%%%%%%%%%%%%%%%%%%%%%%%%%%%%%%%%%%%%%%%%%%%%%%%%%%%%%%%%%%%%%%%%%%%%%%%%%%%%%%%%%%%%%%%%%%%%%%%%%%%%%%%%%%%%%%%%%%%%%%%%%%%%%%%%%%%%%%%%%%%%%%%%%%%%%%%%%%%%%%%%%%%%%%%%%%%%%%%%%%%%%%%%%%%%%%%%%%%%%%%%%%%%%%%%%%%%%%%%%%%%%%%%%%%%%%%%%%%%%%%%%%%%%%%%%%%%%%%%%%%%%%%%%%%%%%%%%%%%%%%%%%%%%%%%%%%%%%%%%%%%%%%%%%%%%%%%%%%%%%%%%%%%%%%%%%%%%%%%%%%%%%%%%%%%

%%%%%%%%%%%%%%%%%%%%%%%%%%%%%%%%%%%%%%%%%%%%%%%%%%%%%%%%%%%%%%%%%%%%%%%%%%%%%%%%%%%%%%%%%%%%%%%%%%%%%%%%%%%%%%%%%%%%%%%%%%%%%%%%%%%%%%%%%%%%%%%%%%%%%%%%%%%%%%%%%%%%%%%%%%%%%%%%%%%%%%%%%%%%%%%%%%%%%%%%%%%%%%%%%%%%%%%%%%%%%%%%%%%%%%%%%%%%%%%%%%%%%%%%%%%%%%%%%%%%%%%%%%%%%%%%%%%%%%%%%%%%%%%%%%%%%%%%%%%%%%%%%%%%%%%%%%%%%%%%%%%%%%%%%%%%%%%%%%%%%%%%%%%%%%%%%%%%%%%%%%%

\end{document}